\documentclass[11pt]{article}

\usepackage[margin = 1in]{geometry}
\usepackage{article_macros}
\usepackage{authblk}
\usepackage{array}

\theoremstyle{definition}
\newtheorem{algorithm}{Algorithm}

\newcommand\numberthis{\addtocounter{equation}{1}\tag{\theequation}}

\newcommand{\cas}{\mathrm{cas}}
\newcommand{\cht}{\mathsf{H}}
\newcommand{\qht}{\mathsf{QHT}}
\newcommand{\cft}{\mathsf{F}}
\newcommand{\qft}{\mathsf{QFT}}
\newcommand{\comph}{\mathsf{cmpIndex}}
\newcommand{\gen}{\mathsf{Gen}}
\newcommand{\ver}{\mathsf{Ver}}

\title{Public-Key Quantum Money and Fast Real Transforms}

\author{Jake Doliskani\thanks{\tt jake.doliskani@mcmaster.ca} }
\author{Morteza Mirzaei\thanks{\tt mirzam48@mcmaster.ca} }
\author{Ali Mousavi\thanks{\tt mousas26@mcmaster.ca} }
\affil{Department of Computing and Software, McMaster University}

\date{}
\allowdisplaybreaks

\begin{document}
\maketitle

\begin{abstract}
    We propose a public-key quantum money scheme based on group actions and the Hartley transform. Our scheme adapts the quantum money scheme of Zhandry (2024), replacing the Fourier transform with the Hartley transform. This substitution ensures the banknotes have real amplitudes rather than complex amplitudes, which could offer both computational and theoretical advantages.

    To support this new construction, we propose a new verification algorithm that uses group action twists to address verification failures caused by the switch to real amplitudes. We also show how to efficiently compute the serial number associated with a money state using a new algorithm based on continuous-time quantum walks. Finally, we present a recursive algorithm for the quantum Hartley transform, achieving lower gate complexity than prior work, and demonstrate how to compute other real quantum transforms, such as the quantum sine transform, using the quantum Hartley transform as a subroutine.
\end{abstract}

\newpage
\section{Introduction}
\label{sec:intro}

In a seminal paper, Wiesner \cite{wiesner1983conjugate} introduced the concept of quantum money, where bills are represented by quantum states. In contrast with classical money, quantum money cannot be counterfeited using general copying machines. This is a consequence of a fundamental theorem in quantum mechanics called the no-cloning theorem. Wiesner's scheme, which in modern terminology is called a private-key quantum money scheme, had significant drawbacks. In particular, in a private-key scheme the bank is required to verify each bill, meaning it must be involved in every transaction. This is the main reason that these schemes are not generally practical.

In 2009, Aaronson \cite{aaronson2009quantum} proposed the first concrete proposal for a different type of quantum money, known as \textit{public-key} quantum money, that did not have such drawbacks. In public-key quantum money anyone can verify the bill while only the bank can issue it. Aaronson's scheme was later broken by Lutomirski et al. \cite{lutomirski2009breaking}. In the years since, several alternative constructions have been explored \cite{aaronson2012quantum, farhi2012quantum, zhandry2021quantum, kane2021quantum, khesin2022publicly, liu2023another, zhandry2024quantum}, yet each has either been broken \cite{conde2019non, roberts2021security, bilyk2023cryptanalysis, liu2023another} or relies on non-standard cryptographic assumptions.

\paragraph{Quantum money from group actions and the Fourier transform.}
A candidate for secure public-key quantum money based on abelian group actions was recently proposed by Zhandry \cite{zhandry2024quantum}. In this scheme, money states correspond to group-action Fourier states, and serial numbers are elements of a group; see Section \ref{sec:qm_Fourier}. Verification is performed using a group-action phase kickback unitary that extracts the serial number from the money state. It was later proved by Doliskani \cite{cryptoeprint:2025/092} that the scheme is secure in the generic group action model.

\subsection{This work}

\paragraph{Motivation.}
We propose a public-key quantum money scheme, which is an adaptation of Zhandry's scheme, using the quantum Hartley transform over finite abelian groups. The motivation behind this proposal is multifold. First, as a consequence of using the Hartley transform, the banknotes will have real amplitudes. This is in contrast to the original scheme, in which the banknotes have complex amplitudes. We believe that this might lead to both computational and theoretical advantages. As a concrete example of a theoretical difference between real and complex quantum states, consider the following scenario: given any two real orthonormal bases $\{\ket{\phi_j}\}$ and $\{\ket{\psi_j}\}$ of a Hilbert space $\X$, one can easily show that
\[ \sum_j \ket{\phi_j} \ket{\phi_j} = \sum_j \ket{\psi_j} \ket{\psi_j}. \]
Such an identity does not, in general, hold for complex bases, and this is the main reason Theorem 4.4 of \cite{zhandry2024quantum} fails to prove that the scheme is a quantum lightning scheme.

Another motivation behind our instantiation is that this work can serve as a starting point for the use of real quantum transforms in quantum cryptography and, more broadly, in quantum computing. To the best of our knowledge, this is the first time the quantum Hartley transform has been successfully used in a significant quantum computing construction. Finally, we hope that this work itself serves as motivation for further research on the optimization of real quantum transforms.

\paragraph{Contributions.}
We propose new algorithms for both the verification of the quantum money scheme and the efficient computation of real quantum transforms.

Replacing the quantum Fourier transform with the quantum Hartley transform in Zhandry's scheme is a straightforward task, as the Hartley transform enjoys essentially the same arithmetic properties as the Fourier transform. However, the verification algorithm under the Hartley transform breaks down, in the sense that it fails to distinguish between certain banknotes. As a result, the verification algorithm accepts some illegitimate banknotes with probability $1$. To address this issue, we propose a new verification algorithm that relies on group action \textit{twists}.

Verification algorithms, in any quantum money scheme, take as input the money state along with its associated serial number. However, we show that the serial number of any given money state can be efficiently computed. To achieve this, we propose a new algorithm based on continuous-time quantum walks.

With regard to efficient real quantum transforms, we present a new algorithm for the quantum Hartley transform that exploits its recursive structure. A similar recursive algorithm was previously proposed in \cite{agaian2002quantum}, based on a well-known decomposition of the Hartley transform. Our algorithm, however, is simpler to implement and easier to analyze, which allows us to provide a more explicit estimate of its gate complexity. We also compare our algorithm to that of \cite{klappenecker2001irresistible}, which uses the quantum Fourier transform as a subroutine. Compared to both these algorithms, we demonstrate that our algorithm achieves lower gate complexity. Finally, we show how to efficiently implement other real quantum transforms using the quantum Hartley transform, by presenting an efficient algorithm for the quantum sine transform that uses the quantum Hartley transform as a subroutine.

\section{Preliminaries}
\label{sec:prelim}

We follow the presentation of \cite{cryptoeprint:2025/092} regarding group actions and quantum computation throughout this paper. A finite Hilbert space $\X$ of dimension $N$ is a complex Euclidean space, isomorphic as a $\C$-vector space to $\C^N$. An $N$-dimensional quantum state is represented by a unit vector $\ket{\psi} \in \mathcal{H}$, where $\ket{\cdot}$ denotes Dirac notation. In practice, we typically use specific bases for a Hilbert space. For example, when considering a finite group $G$, we work with the Hilbert space $\X = \C^G$, which is a $\abs{G}$-dimensional space spanned by the basis $\{\ket{g} : g \in G\}$. 

\subsection{Group actions}

For a group $G$ and a set $X$, we say that $G$ acts on $X$ if there is a mapping $*: G \times X \to X$ that satisfies the following properties:
\begin{enumerate}
    \item Compatibility: for every $g, h \in G$ and every $x \in X$, $g * (h * x) = (gh) * x$,
    \item Identity: for the identity $1 \in G$ and every $x \in X$, $1 * x = x$. 
\end{enumerate}
We use the notation $(G, X, *)$ to denote a group $G$ acting on a set $X$ through the action $*$. A group action is called \textit{regular} if for every $x, y \in X$ there exists a unique $g \in G$ such that $g * x = y$. In order for a group action $(G, X, *)$ to be suitable for algorithmic applications, the group $G$, the set $X$, and the action $*$ must adhere to specific properties. This motivates the concept of an \textit{effective group action}.

Let $G$ be a finite group and $X$ a finite set. A group action $(G, X, *)$ is said to be effective if it satisfies the following properties:
\begin{enumerate*}[label = \roman*)]
    \item There are efficient algorithms for elementary operations such as membership testing, equality testing, sampling uniform elements and group operation for $G$,
    \item There are efficient algorithms for membership testing and unique representation in $X$, and
    \item There exists an efficient algorithm for the action $*$.
\end{enumerate*}

In this paper, we assume that all group actions are effective and, unless otherwise stated, that we are working with regular group actions. A central problem in group-action cryptography is the discrete logarithm problem for group actions, defined as follows:
\begin{definition}[Group Action DLP]
    Let $(G, X, *)$ be an effective group action. Given a pair $(x, g * x)$, where $g \in G$, the discrete logarithm problem (DLP) is to compute $g$.
\end{definition}
We say that the DLP assumption holds for $(G, X, *)$ if no quantum polynomial time (QPT) algorithm can solve the DLP with respect to $(G, X, *)$. A group action for which the DLP assumption holds is referred to as a \textit{cryptographic group action}.

\subsection{The Fourier transform}

Let $G$ be an abelian group. The set of characters of $G$, denoted by $\hat{G}$, is the set of homomorphisms $\chi(a, \cdot): G \to \C$ where $a \in G$. If $G \cong \Z_{N_1} \oplus \cdots \oplus \Z_{N_k}$ then the character $\chi(a, \cdot)$ can be explicitly written as
\[ \chi(a, x) = \omega_{N_1}^{a_1x_1} \cdots \omega_{N_k}^{a_kx_k} \]
where $\omega_M = \exp(2\pi i/ M)$ is a primitive $M$-th root of unity. The Fourier transform of a function $f: G \to \C$ is given by
\[ \hat{f}(a) = \frac{1}{\sqrt{\abs{G}}} \sum_{x \in G} \chi(a, x) f(x). \]
The quantum Fourier transform of a (normalized) state $\sum_{g \in G} f(g) \ket{g}$ is given by $\sum_{x \in G} \hat{f}(x) \ket{x}$. For a regular group action $(G, X, *)$, any subset $S \subseteq G$, any $y \in X$, and any $h \in G$, we define
\begin{equation}
    \label{eq:x-fourier-basis}
    \ket{S^{(h)} * y} = \frac{1}{\sqrt{\abs{S}}} \sum_{g \in S} \chi(g, h) \ket{g * y}.
\end{equation}
There are two orthonormal bases of the space $\C^X$. One basis is $\{ \ket{x} : x \in X \}$. For a fixed element $x \in X$, this basis is the same as $\{\ket{g * x} : g \in G \}$, which follows from the fact that the action is regular and thus $\abs{X} = \abs{G}$. The other basis is given by the states
\[ \ket{G^{(h)} * x} = \frac{1}{\sqrt{\abs{G}}} \sum_{g \in G} \chi(g, h) \ket{g * x}, \quad h \in G. \]
These states are simultaneous eigenstates of the group action operation. Specifically, for the unitary $U_k: \ket{y} \mapsto \ket{k * y}$, where $k \in G$, we have $U_k \ket{G^{(h)} * x} = \chi(-k, h) \ket{G^{(h)} * x}$. These states resemble the set of Fourier states over the abelian group $G$. We will also sometimes refer to them as \textit{Fourier states}.

\paragraph{The $\comph$ algorithm.}
Given a state $\ket{G^{(h)} * x}$, there is an efficient algorithm for computing $h$. Specifically, there is a unitary operator that performs the transformation $\ket{G^{(h)} * x} \ket{0} \mapsto \ket{G^{(h)} * x} \ket{h}$ using the \textit{phase kickback} technique. To see this, start with the state $\ket{G^{(h)} * x} \ket{0}$, apply the quantum Fourier transform to the second register, and then apply the unitary $\sum_{k \in G} U_k \otimes \ket{k}\bra{k}$ to both registers. This results in the state
\[ \frac{1}{\sqrt{\abs{G}}} \sum_{k \in G} \ket{G^{(h)} * x} \chi(-k, h) \ket{k}. \]
Finally, applying the inverse quantum Fourier transform to the second register yields $\ket{G^{(h)} * x} \ket{h}$.

\section{Fast Real Transforms}

\paragraph{The Hartley transform.}
Let $N$ be a positive integer, and let $\Z_N$ be the additive cyclic group of integers modulo $N$. The Hartley transform of a function $f: \Z_N \to \R$ is the function $\cht_N(f): \Z_N \to \R$ defined by
\[ \cht_N(f)(a) = \frac{1}{\sqrt{N}} \sum_{y = 0}^{N - 1} \cas\Big(\frac{2 \pi ay}{N}\Big) f(y),  \]
where $\cas(x) = \cos(x) + \sin(x)$. Like the Fourier transform, $\cht_N$ is a linear operator, and it can be easily shown that it is unitary. The Hartley transform can be defined for general abelian groups. Let $G = \Z_{N_1} \oplus \Z_{N_2} \oplus \cdots \oplus \Z_{N_k}$ be a decomposition of an abelian group $G$ into cyclic groups. For a function $f: G \to \R$, the transform $\cht_G(f): G \to \R$ is defined by
\begin{equation}
    \label{eq:cht-add}
    \cht_G(f)(a_1, \dots, a_k) = \frac{1}{\sqrt{\abs{G}}} \sum_{y_1 = 0}^{N_1 - 1} \cdots \sum_{y_k = 0}^{N_k - 1} \cas\Big(\frac{2 \pi a_1y_1}{N_1} + \cdots + \frac{2 \pi a_ky_k}{N_k}\Big) f(y_1, \dots, y_k),
\end{equation}
where $(a_1, \dots, a_k) \in G$. Let $\alpha(\bm{y}) = y_1 / N_1 + \cdots + y_k / N_k$. A compact form of the expression for $\cht_G$ is
\[ \cht_G(f)(\bm{a}) = \frac{1}{\sqrt{\abs{G}}} \sum_{\bm{y} \in G} \cas(2\pi \lrang{\bm{a}, \alpha(\bm{y})}) f(\bm{y}). \]

It follows from the identity
\[ \cas(2\pi \lrang{\bm{a}, \alpha(\bm{y})}) = \frac{1 - i}{2}\chi(\bm{a}, \bm{y}) + \frac{1 + i}{2}\chi(-\bm{a}, \bm{y}) \]
that
\begin{equation}
    \label{eq:ht-ft}
    \cht_G = \frac{1 - i}{2}\cft_G + \frac{1 + i}{2}\cft_G^*,
\end{equation}
where $\cft_G$ is the Fourier transform over $G$. This proves that $\cht_G^2 = \mathds{1}$, implying that $\cht_G$ is also a unitary transform. Another variant of the Hartley transform used in the literature is a multiplicative variant defined as
\begin{equation}
    \label{eq:cht-mult}
    \cht_G(f)(a_1, \dots, a_k) = \frac{1}{\sqrt{\abs{G}}} \sum_{y_1 = 0}^{N_1 - 1} \cdots \sum_{y_k = 0}^{N_k - 1} \mathrm{cas}\Big(\frac{2 \pi a_1y_1}{N_1}\Big) \cdots \cas\Big(\frac{2 \pi a_ky_k}{N_k}\Big) f(y_1, \dots, y_k).
\end{equation}
This is also a unitary transform as shown by the following lemma.

\begin{lemma}
    The Hartley transform \eqref{eq:cht-mult} is unitary. 
\end{lemma}

\begin{proof}
    If we represent the function $f$ as a vector indexed by the elements of $G$, the matrix representing $\cht_G$ is the tensor product of the matrices representing the Hartley transform over each component $\Z_{N_i}$ of $G$. More precisely, $\cht_G = \cht_{N_1} \otimes \cht_{N_2} \otimes \cdots \otimes \cht_{N_k}$. Since each $\cht_{N_i}$ is unitary, it follows that $\cht_G$ is also unitary.
\end{proof}

The quantum Hartley transform over an abelian group $G$ is defined as
\begin{equation}
    \label{eq:qht}
    \qht_G: \sum_{\bm{x} \in G} f(\bm{x}) \ket{\bm{x}} \mapsto \sum_{\bm{a} \in G} \cht_G(f)(\bm{a}) \ket{\bm{a}},
\end{equation}
where we have skipped the normalization for the left hand side state. Here, $\cht_G(f)$ can be either of the transforms \eqref{eq:cht-add} or \eqref{eq:cht-mult}. For a single basis element of the cyclic group $\Z_N$, the quantum Hartley transform simplifies to
\begin{equation}
	\label{eq:qht-N}
	\qht_N: \ket{a} \mapsto \frac{1}{\sqrt{N}} \sum_{y = 0}^{N - 1} \cas\Big(\frac{2\pi a y}{N}\Big) \ket{y}.
\end{equation}

\paragraph{The sine and cosine transforms.}
Two other widely used real transforms are the discrete sine and cosine transforms. There are different versions of these transforms, but here we briefly introduce the first sine transform, $S_{N - 1}^\text{I}$, and the first cosine transform, $C_{N + 1}^\text{I}$. For more details, see  \cite{rao2014discrete} or Section \ref{sec:fast-sine}. The sine transform $S_{N - 1}^\text{I}$ of a function $f: \Z_N \to \R$ is given by
\[ S_{N - 1}^\text{I}(f)(a) = \left( \frac{2}{N} \right)^{1 / 2} \sum_{y = 1}^{N - 1} \sin\left( \frac{\pi a y}{N} \right) f(y). \]
The quantum version of this transform is defined by
\[ \mathsf{QS}_{N - 1}^\text{I} \ket{a} = \left( \frac{2}{N} \right)^{1 / 2} \sum_{y = 1}^{N - 1} \sin\left( \frac{\pi a y}{N} \right) \ket{y}. \]
The cosine transform of a function $f: \Z_N \to \R$ is given by
\[ C_{N + 1}^\text{I}(f)(a) = \left( \frac{2}{N} \right)^{1 / 2} \sum_{y = 0}^{N} k_a k_y \cos\left( \frac{\pi a y}{N} \right) f(y), \]
where $k_j = 1 / \sqrt{2}$ for $j = 0, N$ and $k_j = 1$ for $j \ne 0, N$. The quantum version of this transform is defined by
\[ \mathsf{QC}_{N + 1}^\text{I} \ket{a} = \left( \frac{2}{N} \right)^{1 / 2} \sum_{y = 0}^{N} k_a k_y \cos\left( \frac{\pi a y}{N} \right) \ket{y}, \]

\subsection{Efficient Quantum Real Transforms}
\label{sec:fast-qht}

The first set of efficient algorithms for computing the quantum Hartley transform ($\qht_N$) and the quantum sine and cosine transforms was introduced in \cite{klappenecker2001discrete, klappenecker2001irresistible}. A recursive algorithm for $\qht_N$ was later proposed in \cite{agaian2002quantum}. In this section, we briefly review the algorithms for $\qht_N$, $\mathsf{QS}_{N - 1}^\text{I}$, and $\mathsf{QC}_{N + 1}^\text{I}$. We refer the reader to \cite{klappenecker2001irresistible} for other versions of the sine and cosine transforms. 

\paragraph{$\qht$ using $\qft$.}
To compute the $\qht_N$, the identity in \eqref{eq:qht-qft} is used as follows. Define the conditional Fourier transform $F = \ket{0}\bra{0} \otimes \mathds{1} + \ket{1}\bra{1} \otimes \qft_N$, and the unitary
\[ R = \frac{1}{2}
    \begin{bmatrix}
        1 - i & 1 + i \\
        1 + i & 1 - i
    \end{bmatrix}.
\]
Given $a \in \Z_N$, the algorithm initializes an ancilla qubit in the zero state and then performs the following operations:
\begin{align*}
    \ket{0}\ket{a}
    & \mapsto \frac{1}{\sqrt{2}} (\ket{0} + \ket{1}) \ket{a} & (H \otimes \mathds{1}) \\
    & \mapsto \frac{1}{\sqrt{2}} (\ket{0} + \ket{1}) \qft_N\ket{a} & (\mathds{1} \otimes \qft_N) \\
    & \mapsto \frac{1}{\sqrt{2}} \ket{0} \qft_N\ket{a} + \frac{1}{\sqrt{2}} \ket{1} \qft_N^3\ket{a} & (F^2) \\
    & \mapsto \frac{1}{\sqrt{2}} \ket{0} \qht_N\ket{a} + \frac{1}{\sqrt{2}} \ket{1} \qft_N^{-2} \qht_N\ket{a} & (R \otimes \mathds{1}) \\
    & \mapsto \frac{1}{\sqrt{2}} (\ket{0} + \ket{1}) \qht_N\ket{a} & (F^2) \\
    & \mapsto \ket{0} \qht_N\ket{a}. & (H \otimes \mathds{1}) \numberthis\label{eq:qht-qft}
\end{align*}
The operation applied at each step is indicated on the right.

\paragraph{Recursive $\qht$.}
The recursive algorithm proposed in \cite{agaian2002quantum} for computing $\qht$, uses the following decomposition \cite{wickerhauser1996adapted}
\begin{equation}
    \label{eq:qht_decomp}
    \qht_N = \frac{1}{\sqrt{2}}
    \begin{bmatrix}
        \mathds{1}_{N / 2} & \mathds{1}_{N / 2} \\
        \mathds{1}_{N / 2} & -\mathds{1}_{N / 2}
    \end{bmatrix}
    \begin{bmatrix}
        \mathds{1}_{N / 2} & 0 \\
        0 & BC_{N / 2}
    \end{bmatrix}
    \begin{bmatrix}
        \qht_{N / 2} & 0 \\
        0 & \qht_{N / 2}
    \end{bmatrix}
    Q_N.
\end{equation}
Both the operators $BC_{N / 2}$ and $Q_N$ can be efficiently implemented as unitary operators.

\paragraph{Quantum sine and cosine.}
To compute $\mathsf{QS}_{N - 1}^\text{I}$ and $\mathsf{QC}_{N + 1}^\text{I}$, a base change unitary $T_N$ is used, which satisfies the identity
\[ T_N^* \cdot \qft_{2N} \cdot T_N = \mathsf{QC}_{N + 1}^\text{I} \oplus i\mathsf{QS}_{N - 1}^\text{I}. \]
The action of $T_N$ is described by
\begin{equation}
    \label{eq:T_N-sine}
    T_N \ket{bx} =
    \begin{cases}
        \ket{bx} & \text{ if } x = 0, \\
        \frac{1}{\sqrt{2}} (\ket{0x} + \ket{1x'}) & \text{ if } x \ne 0, b = 0, \\
        \frac{i}{\sqrt{2}} (\ket{0x} - \ket{1x'}) & \text{ if } x \ne 0, b = 1,
    \end{cases}
\end{equation}
where $x' = N - x$ is the two’s complement of $x$. The unitary $T_N$ can be efficiently implemented using elementary gates. Therefore, the sine and cosine transforms above can be efficiently computed by adding an ancilla qubit and applying $T_N$, $\qft_{2N}$, and $T_N^*$. Note that 
\[ T_N^* \cdot \qft_{2N} \cdot T_N = \ket{0}\bra{0} \otimes \mathsf{QC}_{N + 1}^\text{I} + \ket{1}\bra{1} \otimes i\mathsf{QS}_{N - 1}^\text{I}. \]
Thus, to compute $\mathsf{QC}_{N + 1}^\text{I} \ket{a}$, the state $\ket{0} \ket{a}$ must be prepared, whereas to compute $\mathsf{QS}_{N - 1}^\text{I} \ket{a}$, the state $\ket{1} \ket{a}$ must be prepared.

\section{A New Algorithm for $\qht$}
\label{sec:new-qht}

Our algorithm for the quantum Hartley transform, $\qht_N$, is inspired by the recursive algorithm for computing the quantum Fourier transform ($\qft_N$). Let us briefly explain how the algorithm for $\qft_N$ works. For simplicity, we assume $N = 2^n$, so that elements of $\Z_N$ are represented using exactly $n$ qubits. The generalization to arbitrary $N$ can be achieved by following the same recursive approach. Given $a \in \Z_N$, the expression for $\qft_N$ is written as follows:
\begin{align*}
    \qft_N\ket{a}
    & = \frac{1}{\sqrt{N}} \sum_{y = 0}^{N - 1} \omega_N^{ay}\ket{y} \\
    & = \frac{1}{\sqrt{N}} \sum_{y = 0}^{N / 2 - 1} \omega_N^{ay} \ket{y} + (-1)^a \frac{1}{\sqrt{N}} \sum_{y = 0}^{N / 2 - 1} \omega_N^{ay} \ket{y + N/2} \\
    & = \frac{1}{\sqrt{N / 2}} \sum_{y = 0}^{N / 2 - 1} \omega_N^{ay} \frac{1}{\sqrt{2}} (\ket{0} + (-1)^a \ket{1}) \ket{y}, \numberthis\label{eq:qft_alt}
\end{align*}
where, in the last equation, we have separated the first qubit for clarity. Let $\ket{a} = \ket{t}\ket{b}$, where $b$ is the least significant bit of $a$, so that $a = 2t + b$ for some $t \in \Z_{N / 2}$. Applying $\qft_{N / 2}$ to the first register, we obtain the state
\[ \frac{1}{\sqrt{N / 2}} \sum_{y = 0}^{N / 2 - 1} \omega_N^{2ty} \ket{y} \ket{b}. \]
Next, we apply the phase unitary $P(y, b): \ket{y} \ket{b} \mapsto \omega_N^{by} \ket{y} \ket{b}$, and finally, we apply a Hadamard transform to the last qubit. The result is the state in \eqref{eq:qft_alt}.

We now present our algorithm for the efficient computation of the quantum Hartley transform $\qht_N$. The idea is to exploit the recursive structure of $\qht_N$, similar to the approach we used for $\qft_N$. To this end, we first rewrite the sum in \eqref{eq:qht-N} to reveal its recursive nature. We proceed as follows:
\begin{equation}
    \label{eq:cas-expand}
	\frac{1}{\sqrt{N}} \sum_{y = 0}^{N - 1} \cas\Big( \frac{2\pi a y}{N} \Big) \ket{y}
    = \frac{1}{\sqrt{N}} \sum_{y = 0}^{N / 2 - 1} \cas\Big( \frac{2\pi a y}{N} \Big) \ket{y} + \frac{1}{\sqrt{N}} \sum_{y = N / 2}^{N - 1} \cas\Big( \frac{2\pi a y}{N} \Big) \ket{y}.
\end{equation}
The second sum in the right-hand side can be written as
\begin{align*}
	\sum_{y = N / 2}^{N - 1} \cas\Big( \frac{2\pi a y}{N} \Big) \ket{y}
    & = \sum_{y = 0}^{N / 2 - 1} \cas\Big( \frac{2\pi a y}{N} + \pi a \Big) \ket{y + N/2} \\
    & = (-1)^a \sum_{y = 0}^{N / 2 - 1} \cas\Big( \frac{2\pi a y}{N} \Big) \ket{y + N/2},
\end{align*}
where the second equality follows from the $\cas$ angle-sum identity $\cas(\alpha + \beta) = \cos(\alpha) \cas(\beta) + \sin(\alpha) \cas(-\beta)$, and the fact that $\cos(\pi a) = (-1)^a$ for all $a$. Substituting this into \eqref{eq:cas-expand} gives
\begin{align}
	\frac{1}{\sqrt{N}} \sum_{y = 0}^{N - 1} \cas\Big( \frac{2\pi a y}{N} \Big) \ket{y}
    & = \frac{1}{\sqrt{N}} \sum_{y = 0}^{N / 2 - 1} \cas\Big( \frac{2\pi a y}{N} \Big) (\ket{y} + (-1)^a \ket{y + N/2}) \nonumber \\
    & = \frac{1}{\sqrt{N / 2}} \sum_{y = 0}^{N / 2 - 1} \cas\Big( \frac{2\pi a y}{N} \Big) \frac{1}{\sqrt{2}} (\ket{0} + (-1)^a \ket{1}) \ket{y}, \label{eq:qht-alt} 
\end{align}
where, in the last equality, the most significant qubit has been separated to highlight its role.

We now show how to compute $\qht_N$ recursively. Given $a \in \Z_N$, we once again express $\ket{a} = \ket{t}\ket{b}$, where $b$ is the least significant bit, so that $a = 2t + b$ for some $t \in \Z_{N / 2}$. Suppose we already have an efficient circuit for computing $\qht_{N / 2}$. To compute the Hartley transform of $\ket{a}$, we introduce an ancilla qubit, resulting in the state $\ket{0} \ket{t} \ket{b}$. Then
\begin{align*}
	\ket{0}\ket{t}\ket{b}
    & \mapsto \frac{1}{\sqrt{N / 2}} \sum_{y = 0}^{N / 2 - 1} \cas\Big( \frac{2\pi t y}{N / 2} \Big) \ket{0} \ket{y} \ket{b} & (\mathds{1} \otimes \qht_{N / 2} \otimes \mathds{1}) \\
    & = \frac{1}{\sqrt{N / 2}} \sum_{y = 0}^{N / 2 - 1} \cas\Big( \frac{4\pi t y}{N} \Big) \ket{0} \ket{y} \ket{b} \\
    & \mapsto \frac{1}{\sqrt{N}} \sum_{y = 0}^{N / 2 - 1} \cas\Big( \frac{4\pi t y}{N} \Big) (\ket{0} + \ket{1}) \ket{y}\ket{b}. & (H \otimes \mathds{1})
\end{align*}
Next, we apply the controlled negation
\[
    V: \ket{0}\ket{y} \mapsto \ket{0}\ket{y}, \quad \ket{1}\ket{y} \mapsto \ket{1}\ket{N / 2 -y},
\]
to the first two registers to obtain the state
\[
    \frac{1}{\sqrt{N}} \sum_{y = 0}^{N / 2 - 1} \cas\Big( \frac{4\pi t y}{N} \Big) \ket{0} \ket{y} \ket{b} + \frac{1}{\sqrt{N}} \sum_{y = 0}^{N / 2 - 1} \cas\Big( \frac{4\pi t y}{N} \Big) \ket{1} \ket{N / 2 - y} \ket{b}.
\]
A change of variables in the second sum, along with the fact that $\cas(4\pi t (N / 2 - y) / N) = \cas(-4\pi t y / N)$, results in the state
\[
    \frac{1}{\sqrt{N}} \sum_{y = 0}^{N / 2 - 1} \Big( \cas\Big( \frac{4\pi t y}{N} \Big) \ket{0} + \cas\Big( -\frac{4\pi t y}{N} \Big) \ket{1} \Big) \ket{y} \ket{b}.
\]
Define the single-qubit rotation
\begin{equation}
    \label{eq:sine-rot}
    R(y, b) = 
    \begin{bmatrix}
        \cos(2\pi b y / N) & \sin(2\pi b y / N) \\
        -\sin(2\pi b y / N) & \cos(2\pi b y / N)
    \end{bmatrix},
\end{equation}
and consider the unitary $U_R: \ket{c}\ket{y}\ket{b} \mapsto (R(y, b)\ket{c})\ket{y}\ket{b}$. Applying $U_R$ and the controlled negation $V$, we obtain the state
\begin{align*}
    \ket{\phi_1}
    & = \frac{1}{\sqrt{N}} \sum_{y = 0}^{N / 2 - 1} \cas\Big( \frac{2\pi ay}{N} \Big) \ket{0} \ket{y} \ket{b} + \frac{1}{\sqrt{N}} \Big( \ket{1} \ket{0} + \sum_{y = 0}^{N / 2 - 1} \cas\Big(-\frac{2\pi a(N /2 - y)}{N} \Big) \ket{1} \ket{y} \Big) \ket{b} \\
    & = \frac{1}{\sqrt{N}} \sum_{y = 0}^{N / 2 - 1} \cas\Big( \frac{2\pi ay}{N} \Big) \ket{0} \ket{y} \ket{b} + \frac{1}{\sqrt{N}} \Big( \ket{1} \ket{0} + (-1)^b \sum_{y = 0}^{N / 2 - 1} \cas\Big( \frac{2\pi ay}{N} \Big) \ket{1} \ket{y} \Big) \ket{b},
\end{align*}
where we used the fact that $\cas(-\pi a + 2\pi ay /N ) = (-1)^a \cas(2\pi ay /N )$ and $(-1)^a = (-1)^b$. Consider a controlled-\textsc{cnot} unitary $C_X$ that takes $\ket{c} \ket{y} \ket{b}$ to $\ket{c} \ket{y} \ket{1 \oplus b}$ if $c = 1$ and $y = 0$, and acts as identity otherwise. Then applying the unitary $HC_XH$ to the above state results in the state
\[ \frac{1}{\sqrt{N}} \sum_{y = 0}^{N / 2 - 1} \cas\Big( \frac{2\pi ay}{N} \Big) \Big( \ket{0} + (-1)^b \ket{1} \Big) \ket{y} \ket{b}. \]
Applying the Hadamard transform to the first qubit gives
\[
    \ket{\psi} = \frac{1}{\sqrt{N / 2}} \sum_{y = 0}^{N / 2 - 1} \cas\Big( \frac{2\pi ay}{N} \Big) \ket{b} \ket{y} \ket{b}
\]
Next, we uncompute the first qubit using a \textsc{cnot} gate with the last qubit, then apply a Hadamard transform to obtain:
\begin{align*}
    \ket{\psi}
    & \mapsto \frac{1}{\sqrt{N / 2}} \sum_{y = 0}^{N / 2 - 1} \cas\Big( \frac{2\pi ay}{N} \Big) \ket{0} \ket{y} \ket{b} & (\textsc{cnot}) \\
    & \mapsto \frac{1}{\sqrt{N / 2}} \sum_{y = 0}^{N / 2 - 1} \cas\Big( \frac{2\pi ay}{N} \Big) \ket{0} \ket{y} \frac{1}{\sqrt{2}}(\ket{0} + (-1)^b \ket{1}) & (\mathds{1} \otimes H) \\
    & = \frac{1}{\sqrt{N / 2}} \sum_{y = 0}^{N / 2 - 1} \cas\Big( \frac{2\pi ay}{N} \Big) \ket{0} \ket{y} \frac{1}{\sqrt{2}}(\ket{0} + (-1)^a \ket{1}). & ((-1)^b = (-1)^a) \\
\end{align*}
The final sum matches \eqref{eq:qht-alt}, except that the last qubit appears at the end rather than immediately after the ancilla qubit. This discrepancy can be resolved by relabeling the qubits throughout the recursion and applying a single permutation at the end of the computation. Therefore, we have successfully performed the operation:
\[ \ket{0}\ket{a} \mapsto \ket{0} \qht_N\ket{a}. \]
The following algorithm summarizes the steps described above.

\begin{algorithm}[$\qht_N$] \
    \label{alg:qht-N}
    \begin{description}[font = \normalfont\itshape, itemsep = 0mm, parsep = 0mm, topsep = 1mm]
        \item [Input:] quantum state $\ket{\psi} \in \C^N$, where $N = 2^n$
        \item [Output:] quantum state $\qht_N\ket{\psi}$
    \end{description}

    \begin{enumerate}[itemsep = 0mm, parsep = 1mm, topsep = 1mm]
        \item Initialize an ancilla qubit to $0$ to obtain the state $\ket{0}\ket{\psi}$
        \item Compute $\mathds{1} \otimes \qht_{N / 2} \otimes \mathds{1}$ recursively.
        \item Apply $H \otimes \mathds{1}$.
        \item\label{stp:ngt} Apply the controlled negation $\ket{0}\ket{y} \mapsto \ket{0}\ket{y}, \ket{1}\ket{y} \mapsto \ket{1}\ket{N / 2 -y}$ to the first two registers.
        \item Apply the unitary $U_R$.
        \item Apply the controlled negation of Step \ref{stp:ngt}
        \item\label{stp:mcx} Apply the unitary $HC_XH$
        \item Apply $H \otimes \mathds{1}$
        \item Apply \textsc{cnot} to the first and last qubits.
        \item Apply $\mathds{1} \otimes H$.
        \item\label{stp:swap} Relabel the qubits to implement the effect of the swap $\ket{0} \ket{y} \ket{b} \mapsto \ket{0} \ket{b} \ket{y}$
        \item Trace out the first qubit
    \end{enumerate}
\end{algorithm}

\begin{theorem}
    \label{thm:qht-cost}
    Algorithm \ref{alg:qht-N} is correct and can be implemented using $2\log^2 N + O(\log N)$ elementary gates.
\end{theorem}
\begin{proof}
    The correctness of the algorithm follows from the preceding discussion. Except for the unitary $U_R$, the negation unitary of Step \ref{stp:ngt}, and the multi-controlled \textsc{cnot} of Step \ref{stp:mcx}, all the steps in the algorithm can be implemented using $O(1)$ elementary gates. The negation in Step \ref{stp:ngt} and the \textsc{cnot} in Step \ref{stp:mcx} can be implemented using $\approx \lceil \log N \rceil$ elementary gates. To implement the unitary $U_R$, which involves constructing the conditional operator $R(y, b)$ for arbitrary $y$ and $b$, we utilize the two-qubit operators
    \[ R_j = \ket{0}\bra{0} \otimes \mathds{1} + \ket{1}\bra{1} \otimes
        \begin{bmatrix}
            \cos(2\pi 2^j / N) & \sin(2\pi 2^j / N) \\
            -\sin(2\pi 2^j / N) & \cos(2\pi 2^j / N)
        \end{bmatrix},
    \]
    for $j = 0, 1, \dots, n - 1$. For $b = 0$, $R(y, 0) = \mathds{1}$, and for $b = 1$, $R(y, 1)$ is the product of the $R_j$, where the control qubit is the $j$th qubit of $y$. Therefore, for a $y$ of size $k = \lceil \log y \rceil$, we can construct $U_R$ using $k$ gates from the set $\{R_j\}$.

    Let $T(N)$ be the gate complexity of Algorithm \ref{alg:qht-N} for an input state of dimension $N$. Assuming access to the gates $R_j$ as elementary gates, it follows from above that $T(N) = T(N / 2) + 4\log N + O(1)$. Therefore, $T(N) = 2\log^2 N + O(\log N)$. 
\end{proof}

\subsection{Comparison With Other Algorithms}

In the following, we briefly compare our algorithm for $\qht_N$ to those of \cite{agaian2002quantum} and \cite{klappenecker2001irresistible}. We show that our algorithm has lower gate complexity than both of these algorithms and is conceptually simpler to implement.

The recursive algorithm of \cite{agaian2002quantum} relies on the decomposition given in \eqref{eq:qht_decomp}. We refer the reader to \cite{agaian2002quantum} for the definitions of $Q_N$ and $BC_{N / 2}$, as well as their gate complexities. The permutation $Q_N$ can be implemented using approximately $\lceil \log N \rceil$ elementary gates. The implementation of $BC_{N / 2}$ involves the following operations:
\begin{enumerate}[itemsep = 0mm]
    \item Compute a controlled two's complement of a $(n - 2)$-bit integer
    \item Compute the multi-controlled phase gate $\ket{1x} \mapsto -\ket{1x}$, $\ket{10} \mapsto \ket{10}$.
    \item Apply the unitary $R(y, 1)$ defined in \eqref{eq:sine-rot}.
    \item Compute a controlled two's complement of a $(n - 2)$-bit integer
\end{enumerate}
Each of the above steps requires approximately $\ceil{\log N}$ elementary gates to implement. Let $T(N)$ denote the gate complexity of the algorithm for input size $N$. Then $T(N) = T(N / 2) + 5\log N + O(1)$, which implies that
\[ T(N) = \frac{5}{2} \log^2 N + O(\log N). \]

As outlined at the beginning of Section \ref{sec:fast-qht} in \eqref{eq:qht-qft}, the algorithm of \cite{klappenecker2001irresistible} for the quantum Hartley transform $\qht_N$ uses the quantum Fourier transform $\qft_N$ in a black-box manner. Therefore, to compare the efficiency of that algorithm with Algorithm \ref{alg:qht-N}, we need to determine a concrete gate complexity of the $\qft_N$. The quantum Fourier transform is well-known and has been extensively studied in the literature. We consider two conventional implementations for $\qft_N$: the first is the one outlined at the beginning of Section \ref{sec:new-qht}, and the second is an algorithm based on decomposing $\qft_N$ into a tensor product of phase operators. 

For the recursive algorithm outlined at the beginning of Section \ref{sec:new-qht}, the dominating step is the application of the unitary $P(y, b)$. Similar to the case of $\qht_N$, the operator $P$ can be implemented using products of phase gates $P_j = \ket{0} \bra{0} + \exp(2\pi 2^j / N) \ket{1} \bra{1}$, for $j = 0, 1, \dots, n - 1$. Assuming access to the $P_j$ as elementary gates, this algorithm results in a gate complexity of $\frac{1}{2}\log^2 N + O(\log N)$ for $\qft_N$. 

The other algorithm for $\qft_N$ is based on the identity
\[ \qft_N \ket{a} = \bigotimes_{j = 1}^n \frac{1}{\sqrt{2}} (\ket{0} + \omega_N^{2^{n - j}a}\ket{1}). \]
Implementing $\qft_N$ using this identity involves applying a Hadamard transform followed by a product of controlled phase transforms $Q_j = \ket{0} \bra{0} \otimes \mathds{1} + \ket{1} \bra{1} \otimes P_j$ to each qubit of $\ket{a}$ \cite{cleve1998quantum}. Again, assuming access to the $P_j$'s as elementary gates, this algorithm also results in a gate complexity of $\frac{1}{2}\log^2 N + O(\log N)$. 

Since the algorithm in \eqref{eq:qht-qft} invokes $\qft_N$ a total of five times, the gate complexity of the algorithm for computing $\qht_N$ is
\[ T(N) = \frac{5}{2} \log^2 N + O(\log N). \]
It follows from the above and Theorem \ref{thm:qht-cost} that our algorithm for $\qht_N$ has a gate complexity that is approximately $1.25 \times$ lower than that of \cite{agaian2002quantum} and \cite{klappenecker2001irresistible}.

\begin{remark}
    It is worth noting that all of the algorithms above are exact and do not require significant extra memory. If significant extra memory is allowed or if only an approximation of $\qft_N$ is needed, there are more efficient algorithms available; see, for example, \cite{hales2000improved, cleve2000fast}.
\end{remark}

\section{Other Real Transforms}
\label{sec:fast-sine}

In this section, we show how the Hartley transform can be used as a subroutine to efficiently implement the quantum sine and cosine transforms. For completeness, we list the various versions of these transforms below.
\begin{center}
    \newcolumntype{M}{>{$\displaystyle}l<{$}}
    \setlength{\tabcolsep}{1mm}
    \begin{tabular}{MM}
        S_{N - 1}^\text{I} & = \left( \frac{2}{N} \right)^{1 / 2} \left[ \sin\left( \frac{m n \pi}{N} \right)  \right], \quad m, n = 1, 2, \dots, N - 1 \\
        S_N^\text{II} & = \left( \frac{2}{N} \right)^{1 / 2} \left[ k_m \sin\left( \frac{m (n - 1 / 2) \pi}{N} \right)  \right], \quad m, n = 1, 2, \dots, N \\    
        S_N^\text{III} & = \left( \frac{2}{N} \right)^{1 / 2} \left[ k_n \sin\left( \frac{(m - 1 / 2) n \pi}{N} \right)  \right], \quad m, n = 1, 2, \dots, N \\    
        S_N^\text{IV} & = \left( \frac{2}{N} \right)^{1 / 2} \left[ \sin\left( \frac{(m + 1 / 2) (n + 1 / 2) \pi}{N} \right)  \right], \quad m, n = 0, 1, \dots, N - 1, \\    
    \end{tabular}
\end{center}
where $k_j = 1 / \sqrt{2}$ for $j = N$ and $k_j = 1$ for $j \ne N$. Different versions of the discrete cosine transform are
\begin{center}
    \newcolumntype{M}{>{$\displaystyle}l<{$}}
    \setlength{\tabcolsep}{1mm}
    \begin{tabular}{MM}
        C_{N + 1}^\text{I} & = \left( \frac{2}{N} \right)^{1 / 2} \left[ k_m k_n \cos\left( \frac{m n \pi}{N} \right)  \right], \quad m, n = 0, 1, \dots, N \\
        C_N^\text{II} & = \left( \frac{2}{N} \right)^{1 / 2} \left[ k_m \cos\left( \frac{m (n + 1 / 2) \pi}{N} \right)  \right], \quad m, n = 0, 1, \dots, N - 1 \\    
        C_N^\text{III} & = \left( \frac{2}{N} \right)^{1 / 2} \left[ k_n \cos\left( \frac{(m + 1 / 2) n \pi}{N} \right)  \right], \quad m, n = 0, 1, \dots, N - 1 \\    
        C_N^\text{IV} & = \left( \frac{2}{N} \right)^{1 / 2} \left[ \cos\left( \frac{(m + 1 / 2) (n + 1 / 2) \pi}{N} \right)  \right], \quad m, n = 0, 1, \dots, N - 1, \\    
    \end{tabular}
\end{center}
where $k_j = 1 / \sqrt{2}$ for $j = 0, N$ and $k_j = 1$ for $j \ne 0, N$.

Two possible approaches for computing the quantum sine and cosine transforms are: employing a recursive method similar to the one in Section \ref{sec:fast-qht}, or using the method proposed in \cite{klappenecker2001irresistible}, where the quantum Fourier transform is used as a subroutine. In our case, the latter approach would involve using the quantum Hartley transform as a subroutine. Since the method outlined in Section \ref{sec:fast-qht} does not appear to offer any significant computational advantage over directly employing $\qht$ as a subroutine, we adopt the latter approach for designing quantum algorithms for the sine and cosine transforms. Specifically, the recursive method seems to have a circuit complexity roughly equivalent to that of $\qht$, whereas, as we will demonstrate below, the latter approach requires exactly one call to $\qht$ and a few elementary operations.

In the following, we adopt the technique in \cite{klappenecker2001irresistible} to design an algorithm for the quantum sine transform $\mathsf{QS}_{N - 1}^\text{I}$. The other transforms can be implemented using a similar method. As mentioned above, to implement $\mathsf{QS}_{N - 1}^\text{I}$, the authors of \cite{klappenecker2001irresistible} introduced a transform $T_N$ such that
\[ T_N^* \cdot \qft_{2N} \cdot T_N = \mathsf{QC}_{N + 1}^\text{I} \oplus i\mathsf{QS}_{N - 1}^\text{I}. \]
The transform $T_N$ can be efficiently implemented using elementary operations. Therefore, the cost of computing $\mathsf{QS}_{N - 1}^\text{I}$ is approximately the same as that of $\qft_{2N}$.

Here, we adapt the transform $T_N$ to work with $\qht_{2N}$. The algorithm proceeds as follows: Given the basis state $\ket{a}$, where $a \in \Z_N$, append an ancilla qubit in the zero state to obtain $\ket{0} \ket{a}$. Applying the transform $HX$ to the first qubit results in
\[ \frac{1}{\sqrt{2}} (\ket{0} \ket{a} - \ket{1} \ket{a}). \]
Next, consider the unitary $U = \ket{0} \bra{0} \otimes \mathds{1} + \ket{1} \bra{1} \otimes \ket{N - x} \bra{x}$. Applying $U_N$ to the above state gives
\[ \frac{1}{\sqrt{2}} (\ket{0} \ket{a} - \ket{1} \ket{N - a}). \]
Define $T_N = U_N (H \otimes \mathds{1}) (X \otimes \mathds{1})$. Considering this state as an $(n + 1)$-qubit state where the leftmost qubit is the most significant qubit, we can apply $\qht_{2N}$ to obtain the state
\begin{align*}
    \ket{\psi}
    & = \frac{1}{2\sqrt{N}} \sum_{y = 0}^{2N - 1} \Big(\cas\Big( \frac{\pi a y}{N} \Big) - \cas\Big( \frac{\pi (2N - a) y}{N} \Big)\Big) \ket{y} \\
    & = \frac{1}{2\sqrt{N}} \sum_{y = 0}^{2N - 1} \Big(\cas\Big( \frac{\pi a y}{N} \Big) - \cas\Big( \frac{-\pi a y}{N} \Big)\Big) \ket{y} \\
    & = \frac{1}{\sqrt{N}} \sum_{y = 1}^{2N - 1} \sin\Big( \frac{\pi a y}{N} \Big) \ket{y}
\end{align*}
The last summation can be written as
\begin{align*}
    \ket{\psi}
    & = \frac{1}{\sqrt{N}} \sum_{y = 1}^{N - 1} \sin\Big( \frac{\pi a y}{N} \Big) \ket{y} + \frac{1}{\sqrt{N}} \sum_{y = N}^{2N - 1} \sin\Big( \frac{\pi a y}{N} \Big) \ket{y} \\
    & = \frac{1}{\sqrt{N}} \sum_{y = 1}^{N - 1} \sin\Big( \frac{\pi a y}{N} \Big) \ket{y} + \frac{1}{\sqrt{N}} \sum_{y = 1}^{N - 1} \sin\Big( \frac{\pi a (2N - y)}{N} \Big) \ket{2N - y} \\
    & = \frac{1}{\sqrt{N}} \sum_{y = 1}^{N - 1} \sin\Big( \frac{\pi a y}{N} \Big) (\ket{y} - \ket{2N - y})
\end{align*}
Separating the most significant qubit, we obtain the state
\[ \ket{\psi} = \left( \frac{2}{N} \right)^{1 / 2} \sum_{y = 1}^{N - 1} \sin\Big( \frac{\pi a y}{N} \Big) \frac{1}{\sqrt{2}} (\ket{0} \ket{y} - \ket{1} \ket{N - y}). \]
Now, we apply $T_N^*$ to obtain
\[ \ket{0} \left( \frac{2}{N} \right)^{1 / 2} \sum_{y = 1}^{N - 1} \sin\Big( \frac{\pi a y}{N} \Big) \ket{y}, \]
To summarize, we have constructed a unitary $T_N$ such that
\[ (T_N^* \cdot \qht_{2N} \cdot T_N) \ket{0} \ket{a} = \ket{0} \mathsf{QS}_{N - 1}^\text{I} \ket{a}. \]

\begin{theorem}
    \label{thm:fast-sine}
    The quantum sine transform $\mathsf{QS}_{N - 1}^\text{I}$ can be implemented using $2\log^2 N + O(\log N)$ elementary gates.
\end{theorem}
\begin{proof}
    The unitary $T_N$ can be implemented using $O(\log N)$ elementary gates. Since the algorithm involves a single call to $\qht_{2N}$, the result follows directly from Theorem \ref{thm:qht-cost}.
\end{proof}


\section{Public-Key Quantum Money}

We devote the rest of the paper to constructing a public-key quantum money scheme based on the Hartley transform. In this section, we review the quantum money scheme proposed by Zhandry \cite{zhandry2024quantum}, which is based on abelian group actions and the quantum Fourier transform. To construct a quantum money scheme based on $\qht$, the idea is to simply replace the $\qft$ with $\qht$ in Zhandry's scheme.

\subsection{Quantum Money From Group Actions}
\label{sec:qm_Fourier}

A public-key quantum money scheme consists of two QPT algorithms: 
\begin{itemize}
    \item $\gen(1^\lambda)$: This algorithm takes a security parameter $\lambda$ as input and outputs a pair $(s, \rho_s)$, where $s$ is a binary string called the serial number, and $\rho_s$ is a quantum state called the banknote. The pair $(s, \rho_s)$, or simply $\rho_s$, is sometimes denoted by $\$$.
    \item $\ver(s, \rho_s)$: This algorithm takes a serial number and an alleged banknote as input and outputs either $1$ (accept) or $0$ (reject).
\end{itemize}

The quantum money scheme is said to be \textit{correct} if genuine banknotes generated by $\gen$ are accepted by $\ver$ with high probability. More formally:
\[ \Pr[\ver(s, \rho_s) = 1 : (s, \rho_s) \gets \gen(1^\lambda)] \ge 1 - \negl(\lambda). \]
where the probability is taken over the randomness of $\gen$ and $\ver$. The scheme $(\gen, \ver)$ is said to be secure if, given a genuine bill $(s, \rho_s)$, no QPT algorithm $\mathcal{A}$ can produce two (possibly entangled) bills $(s, \rho_1)$ and $(s, \rho_2)$ that are both accepted by $\ver$ with non-negligible probability. More formally:
\[ \Pr\left[ \ver(s, \rho_1) = \ver(s, \rho_2) = 1 : \substack{(s, \rho_s) \gets \gen(1^\lambda) \\ (\rho_1, \rho_2) \gets \mathcal{A}(s, \rho_s)} \right] \le \negl(\lambda). \]

We now briefly outline the quantum money construction from \cite{zhandry2024quantum}, which is based on abelian group actions. Let $\{(G_\lambda, X_\lambda, *)\}_{\lambda \in J}$, where $J \subset \N$, be a collection of cryptographic group actions for abelian groups $G_\lambda$, and let $x_\lambda \in X_\lambda$ be a fixed element. The $\gen$ and $\ver$ algorithms are as follows:
\begin{itemize}
    \item $\gen(1^\lambda)$. Begin with the state $\ket{0}\ket{x_\lambda}$, and apply the quantum Fourier transform over $G_\lambda$ to the first register producing the superposition
    \[ \frac{1}{\sqrt{\abs{G_\lambda}}} \sum_{g \in G_\lambda} \ket{g}\ket{x_\lambda}. \]
    Next, apply the unitary transformation $\ket{h}\ket{y} \mapsto \ket{h}\ket{h * y}$ to this state, followed by the quantum Fourier transform on the first register. This results in
    \[ \frac{1}{\abs{G_\lambda}} \sum_{h \in G_\lambda} \sum_{g \in G_\lambda} \chi(g, h) \ket{h}\ket{g * x_\lambda} = \frac{1}{\sqrt{\abs{G_\lambda}}} \sum_{h \in G_\lambda} \ket{h} \ket{G^{(h)} * x_\lambda} \]
    where $\ket{G^{(h)} * x_\lambda}$ is defined as in \eqref{eq:x-fourier-basis}. Measure the first register to obtain a random $h \in G_\lambda$, collapsing the state to $\ket{G^{(h)} * x_\lambda}$. Return the pair $(h, \ket{G^{(h)} * x_\lambda})$.

    \item $\ver(h, \ket{\psi})$. First, check whether $\ket{\psi}$ has support in $X_\lambda$. If not, return $0$. Then, apply $\comph$ to the state $\ket{\psi}\ket{0}$, and measure the second register to obtain some $h' \in G_\lambda$. If $h' = h$, return $1$; otherwise return $0$.
\end{itemize}

From this point forward, to simplify the notation, we make the security parameter $\lambda$ implicit, and use $G$ for $G_\lambda$, $X$ for $X_\lambda$, and so on.

\section{Quantum Money With The Hartley Transform}
\label{sec:qm_hartley}

The quantum money scheme above can be instantiated using the quantum Hartley transform instead of the quantum Fourier transform. However, this substitution breaks the verification algorithm. In the next sections, we will show how quantum walks can address this issue. To understand where the problem arises, we first present the $\gen$ and $\ver$ algorithms for the money scheme, similar to the previous description but with $\qht_G$ replacing $\qft_G$. For simplicity, we assume $G = \Z_N$. Let $x \in \Z_N$ be a fixed element.

\begin{itemize}
    \item $\gen$. Begin with the state $\ket{0}\ket{x}$, and apply the quantum Hartley transform over $\Z_N$ to the first register producing the superposition
    \[ \frac{1}{\sqrt{N}} \sum_{g \in \Z_N} \ket{g}\ket{x}. \]
    Next, apply the unitary $\ket{h}\ket{y} \mapsto \ket{h}\ket{h * y}$ to this state, followed by a $\qht_N$ on the first register. This results in
    \[ \frac{1}{N} \sum_{h \in \Z_N} \sum_{g \in \Z_N} \cas\Big(\frac{2\pi gh}{N}\Big) \ket{h}\ket{g * x} = \frac{1}{\sqrt{N}} \sum_{h \in \Z_N} \ket{h} \ket{\Z_N^{(h)} * x}_H \]
    where
    \[ \ket{\Z_N^{(h)} * x}_H = \frac{1}{\sqrt{N}} \sum_{g \in \Z_N} \cas\Big(\frac{2\pi gh}{N}\Big) \ket{g * x}. \]
    Measure the first register to obtain a random $h \in \Z_N$, collapsing the state to $\ket{\Z_N^{(h)} * x}_H$. Return the pair $(h, \ket{\Z_N^{(h)} * x}_H)$.

    \item $\ver(h, \ket{\psi})$. First, check whether $\ket{\psi}$ has support in $X$. If not, return $0$. Then, apply $\comph$ to the state $\ket{0} \ket{\psi}$, and measure the first register to obtain some $h' \in \Z_N$. If $h' = h$, return $1$; otherwise return $0$.
\end{itemize}

Let us take a closer look at the verification algorithm. Suppose $\ket{\psi}$ is a genuine money state, say $\ket{\psi} = \ket{\Z_N^{(h)} * x}_H$. The $\comph$ algorithm begins by preparing the state
\[ \ket{\phi} = \frac{1}{\sqrt{N}} \sum_{u \in \Z_N} \ket{u} \ket{\Z_N^{(h)} * x}_H. \]
It then applies the unitary $\ket{y} \ket{u} \mapsto \ket{(-u) * y} \ket{u}$, yielding:
\begin{align*}
    \ket{\phi}
    & \mapsto \frac{1}{N} \sum_{u \in \Z_N} \sum_{g \in \Z_N} \cas\Big(\frac{2\pi gh}{N}\Big) \ket{u} \ket{(g - u) * x} \\
    & = \frac{1}{N} \sum_{u \in \Z_N} \sum_{g \in \Z_N} \cas\Big(\frac{2\pi (g + u)h}{N}\Big) \ket{u} \ket{g * x}.
\end{align*}
Finally, the algorithm applies $\qht_N$ to the second register. A straightforward calculation shows that the resulting state is
\[ \frac{1}{\sqrt{N}} \sum_{g \in \Z_N} \Big( \cos\Big( \frac{2\pi gh}{N} \Big) \ket{h} + \sin\Big( \frac{2\pi gh}{N} \Big) \ket{-h} \Big) \ket{g * x}. \]
Rewriting this state using the $\cas$ function, we get the state
\begin{align*}
    \ket{\phi_1}
    & = \frac{1}{\sqrt{N}} \sum_{g \in \Z_N} \frac{1}{\sqrt{2}}\Big( \cas\Big( \frac{2\pi gh}{N} \Big) \ket{h_+} + \cas\Big( -\frac{2\pi gh}{N} \Big) \ket{h_-} \Big) \ket{g * x} \\
    & = \frac{1}{\sqrt{2}} \ket{h_+} \ket{\Z_N^{(h)} * x}_H + \frac{1}{\sqrt{2}} \ket{h_-} \ket{\Z_N^{(-h)} * x}_H, \numberthis\label{eq:ver-sign}
\end{align*}
where $\ket{h_{\pm}} = (\ket{h} \pm \ket{-h}) / \sqrt{2}$. This stands in clear contrast to the case when the scheme is instantiated using the quantum Fourier transform. In that case, we obtain the state $\ket{h} \ket{\Z_N^{(h)} * x}$, from which $h$ can be read off.

Unfortunately, we do not know how to use the resulting state in \eqref{eq:ver-sign} to verify the banknote $(h, \ket{\Z_N^{(h)} * x}_H)$. In particular, when given the banknote $(h, \ket{\Z_N^{(-h)} * x}_H)$, the above verification algorithm produces the state
\begin{equation}
    \label{eq:ver-sign1}
    \frac{1}{\sqrt{2}} \ket{h_+} \ket{\Z_N^{(-h)} * x}_H - \frac{1}{\sqrt{2}} \ket{h_-} \ket{\Z_N^{(h)} * x}_H,
\end{equation}
and with the current assumptions on the group action, we lack a method to distinguish between this state and the one in \eqref{eq:ver-sign}. This limitation motivates us to consider an additional, commonly used property of the group action, known as \textit{twists}.

\subsection{Group Actions with Twists}

We show how twists can be used to successfully verify banknotes in the above money scheme. A group action $(G, X, *)$, with a fixed element $x \in X$, is said to support twists if there exists an efficient algorithm that computes the mapping $g * x \mapsto (-g) * x$ for any $g \in G$. The twisting map can be efficiently performed in isogeny-based group actions, which are the most widely used cryptographic group actions. We assume that the group actions used to construct the above quantum money scheme support twists.

In the quantum setting, the twist operation is the unitary $\textsc{twist}: \ket{g * x} \mapsto \ket{(-g) * x}$. The following lemma states the result of applying the twist operation to a genuine money state.
\begin{lemma}
    For any $h \in \Z_N$, the twist unitary maps $\ket{\Z_N^{(h)} * x}_H$ to $\ket{\Z_N^{(-h)} * x}_H$.
\end{lemma}
\begin{proof}
    We have
    \begin{align*}
        \textsc{twist}\, \ket{\Z_N^{(h)} * x}_H
        & = \frac{1}{\sqrt{N}} \sum_{g \in \Z_N} \cas\Big( \frac{2\pi gh}{N} \Big) \ket{(-g) * x} \\
        & = \frac{1}{\sqrt{N}} \sum_{g \in \Z_N} \cas\Big( \frac{-2\pi gh}{N} \Big) \ket{g * x} \\
        & = \ket{\Z_N^{(-h)} * x}_H. \qedhere
    \end{align*}
\end{proof}

As explained at the end of Section \ref{sec:qm_hartley}, the verification algorithm does not distinguish between the states \eqref{eq:ver-sign} and \eqref{eq:ver-sign1}. Distinguishing these two states reduces to distinguishing between the money states $\ket{\Z_N^{(h)} * x}_H$ and $\ket{\Z_N^{(-h)} * x}_H$. The following lemma shows that this can be efficiently done using the twist unitary.
\begin{lemma}
    \label{lem:dist}
    Given $h \in \Z_N \setminus \{0, N / 4, N / 2, 3N / 4\}$, there exists a polynomial-time quantum algorithm that distinguishes between the states $\ket{\Z_N^{(h)} * x}_H$ and $\ket{0}\ket{\Z_N^{(-h)} * x}_H$. Specifically, the algorithm is a unitary that acts as the identity on $\ket{0} \ket{\Z_N^{(h)} * x}_H$ but maps $\ket{0} \ket{\Z_N^{(-h)} * x}_H$ to $\ket{1} \ket{\Z_N^{(-h)} * x}_H$.
\end{lemma}
\begin{proof}
    The algorithm proceeds as follows. First, we find $u \in \Z_N$ such that $uh = N / 8 \bmod N$; based on the condition on $h$, such a $u$ always exists. For such a $u$, we have
    \begin{align*}
        \cas\Big( \frac{2\pi (g + u) h}{N} \Big)
        & = \cos\Big( \frac{2\pi gh}{N} \Big) \cas\Big( \frac{2\pi uh}{N} \Big) + \sin\Big( \frac{2\pi gh}{N} \Big) \cas\Big(-\frac{2\pi uh}{N} \Big) \\
        & = \cos\Big( \frac{2\pi gh}{N} \Big) \cas\Big( \frac{\pi}{4} \Big) + \sin\Big( \frac{2\pi gh}{N} \Big) \cas\Big( -\frac{\pi}{4} \Big) \\
        & = \sqrt{2} \cos\Big( \frac{2\pi gh}{N} \Big). \numberthis\label{eq:cas-cos}
    \end{align*}
    Similarly,
    \begin{equation}
        \label{eq:cast-sin}
        \cas\Big( \frac{-2\pi (g + u) h}{N} \Big) = -\sqrt{2} \sin\Big( \frac{2\pi gh}{N} \Big).
    \end{equation}
    Now, consider the unitary $T_u: \ket{y} \mapsto \ket{(-u) * y}$. It follows from the above identities that
    \begin{equation}
        \label{eq:cos}
        T_u \ket{\Z_N^{(h)} * x}_H \propto \sum_{g \in \Z_N} \cos\Big( \frac{2\pi gh}{N} \Big) \ket{g * x},
    \end{equation}
    whereas
    \begin{equation}
        \label{eq:sin}
        T_u \ket{\Z_N^{(-h)} * x}_H \propto \sum_{g \in \Z_N} \sin\Big( \frac{2\pi gh}{N} \Big) \ket{g * x}.
    \end{equation}
    Therefore, applying $T_u$ results in one of these states, depending on the input state. Both of these states are eigenstates of the twist unitary. The former corresponds to eigenvalue $1$ and the latter corresponds to eigenvalue $-1$.

    Now, prepare $\ket{0} \ket{\psi}$, where $\ket{\psi}$ is the input state. Applying the sequence of unitaries $\mathds{1} \otimes T_u$, $H \otimes \mathds{1}$, $\ket{0}\bra{0} \otimes \mathds{1} + \ket{1}\bra{1} \otimes \textsc{twist}$, $H \otimes \mathds{1}$, and $\mathds{1} \otimes T_u$ to the state $\ket{0} \ket{\psi}$ results in the state $\ket{0} \ket{\psi}$ (resp. $\ket{1} \ket{\psi}$) if $\ket{\psi}$ is $\ket{\Z_N^{(h)} * x}_H$ (resp. $\ket{\Z_N^{(-h)} * x}_H$). 
\end{proof}

\begin{remark}
    For a general $N$, Lemma \ref{lem:dist} holds with a slight modification. Specifically, we need to find $u$ such that $uh = \lfloor N / 8 \rfloor$. In this case, $\cas(\pi / 4) \approx 1$ and $\cas(-\pi / 4) \approx 0$, with exponentially small error, so the identities \eqref{eq:cas-cos} and \eqref{eq:cast-sin} also hold up to exponentially small error as well. To ensure the existence of such $u$, we could, for example, assume that $h$ is coprime with $N$. Since a random $h$ is coprime with $N$ with overwhelming probability, this condition can be efficiently enforced in the $\gen$ algorithm by repeating the procedure until we find an appropriate $h$.
\end{remark}

\noindent We now give the new verification algorithm that utilizes the twist operation.
\begin{algorithm}[$\ver_{new}$] \
    \label{alg:new-ver}
    \begin{description}[font = \normalfont\itshape, itemsep = 0mm, parsep = 0mm, topsep = 1mm]
        \item [Input:] Alleged banknote $(h, \ket{\psi})$ where $\ket{\psi} \in \C^N$.
        \item [Output:] $0$ or $1$
    \end{description}

    \begin{enumerate}[itemsep = 0mm, parsep = 1mm, topsep = 1mm]
        \item\label{stp:sup} If the support of $\ket{\psi}$ is not in $X$, return $0$.
        \item Apply the $\comph$ algorithm to the state $\ket{0} \ket{\psi}$.
        \item\label{stp:hpm} Measure the first register using the observable
        \[ \{M_0 = \ket{h}\bra{h} + \ket{-h}\bra{-h}, M_1 = \mathds{1} - M_0\}, \]
        \item\label{stp:h-check} If the measurement outcome is $1$, return $0$.
        \item Apply the inverse of $\comph$
        \item Use Lemma \ref{lem:dist} to check whether the resulting state is $\ket{\Z_N^{(h)} * x}$. If yes, return $1$; otherwise, return $0$.
    \end{enumerate}
\end{algorithm}

\noindent Some explanations about the algorithm are in order. Step \ref{stp:sup} can be implemented using an extra qubit initialized to $\ket{0}$, which is set to $1$ if the values of the first register is not in $X$. This qubit is then measured, and the state is rejected if the measurement outcome is not $0$. The post-measurement state, which we denote again by $\ket{\psi}$, has support in $X$. Similarly, implementing the measurement in Step \ref{stp:hpm} can be done efficiently, as $h$ is known. Finally, the last step is implemented using an extra qubit, as stated in Lemma \ref{lem:dist}.

\begin{theorem}
    \label{thm:new-ver}
    Given a banknote $(h, \ket{\psi})$, we have
    \[ \Pr[\ver_{new}(h, \ket{\psi}) = 1] = \abs{\braket{\psi}{\Z_N^{(h)} * x}_H}^2. \]
    Algorithm \ref{alg:new-ver} runs in $\poly(\log N)$ operations. Moreover, if $\ver_{new}$ accepts, then the post-verification state is exactly $\ket{\Z_N^{(h)} * x}_H$.
\end{theorem}
\begin{proof}
    It follows from the prior discussion that the running-time of the algorithm is $\poly(\log N)$. We now prove correctness. Suppose the state $\ket{\psi}$ passes the check in Step \ref{stp:sup}, so that it has support in $X$. Since the states $\ket{\Z_N^{(u)} * x}, u = 0, \dots, N - 1$, form an orthonormal basis of the space $\C^X$, we can write
    \[ \ket{\psi} = \sum_{u \in \Z_N} \alpha_u \ket{\Z_N^{(u)} * x}_H, \]
    where $\sum_u \abs{\alpha_u}^2 = 1$. By the same calculations leading up to \eqref{eq:ver-sign}, applying $\comph$ to the state $\ket{0} \ket{\psi}$ results in the state
    \[ \frac{1}{\sqrt{2}} \sum_{u \in \Z_N} \alpha_u (\ket{u_+} \ket{\Z_N^{(u)} * x}_H + \ket{u_-} \ket{\Z_N^{(-u)} * x}_H). \]
    Steps \ref{stp:hpm} and \ref{stp:h-check} projects this state onto
    \[ \alpha_h (\ket{h_+} \ket{\Z_N^{(h)} * x}_H + \ket{h_-} \ket{\Z_N^{(-h)} * x}_H) + \alpha_{-h} (\ket{h_+} \ket{\Z_N^{(-h)} * x}_H - \ket{h_-} \ket{\Z_N^{(h)} * x}_H), \]
    where we have ignored the normalization constant for clarity. This projection occurs with probability $\abs{\alpha_h}^2 + \abs{\alpha_{-h}}^2$. Applying the inverse of $\comph$, we obtain the state
    \[ \frac{1}{\sqrt{\abs{\alpha_h}^2 + \abs{\alpha_{-h}^2}}} (\alpha_h \ket{\Z_N^{(h)} * x}_H + \alpha_{-h} \ket{\Z_N^{(-h)} * x}_H). \]
    The final step of the algorithm projects this state onto $\ket{\Z_N^{(h)} * x}_H$ with probability $\abs{\alpha_h^2} / (\abs{\alpha_h}^2 + \abs{\alpha_{-h}^2})$. Therefore, the probability of the state $\ket{\psi}$ passing verification is
    \[ (\abs{\alpha_h}^2 + \abs{\alpha_{-h}^2}) \frac{\abs{\alpha_h}^2}{\abs{\alpha_h}^2 + \abs{\alpha_{-h}^2}} = \abs{\alpha_h}^2 = \abs{\braket{\psi}{\Z_N^{(h)} * x}_H}^2. \qedhere \]
\end{proof}

One intriguing question regarding the quantum money scheme above is whether a more ``direct'' verification algorithm can be designed. In the original scheme, using the quantum Fourier transform, we could directly obtain $h$ from the money state $\ket{\Z_N^{(h)} * x}$ and compare it to the given $h$. However, this approach did not work when we used the Hartley transform: we ended up with the state in \eqref{eq:ver-sign}, from which it is unclear how to extract $h$ without collapsing the state. To address this, we design an algorithm for computing $h$ that utilizes quantum walks.


\section{Computing The Serial Number Using Quantum Walks}

In this section, we show how to use continuous-time quantum walks to compute the serial number of a given money state. Quantum walks are quantum analogs of classical random walks and play a fundamental role in quantum algorithms. Similar to the classical case, there are two types of quantum walks: continuous-time and discrete-time, both of which exhibit significantly different behaviors compared to classical random walks. In particular, a key distinction between quantum and classical random walks lies in the wave nature of quantum mechanics. In a quantum walk, interference effects can occur, allowing computational advantage for exploration of the graph compared to classical walks, e.g., \cite{aharonov2001quantum, ambainis2007quantum, childs2003exponential, childs2007quantum}.

For a graph $\Gamma$, the dynamics of a continuous-time classical walk on $\Gamma$ is described by the differential equation $\frac{d}{dt} q(t) = Lq(t)$, where $L$ is the Laplacian of $\Gamma$ and $q(t)$ describes the state of the walk at time $t$. In the quantum setting, the vector $q(t)$ is replaced by a quantum state $\ket{\psi(t)}$ and the dynamics of the walk is given by the Schr\"{o}dinger equation
\begin{equation}
    \label{equ:Schrodinger}
    i\frac{d}{dt}\ket{\psi(t)} = L\ket{\psi(t)},
\end{equation}
where $L$ plays the role of the Hamiltonian of the quantum system. The solution to this differential equation can be written in closed form as:
\[ \ket{\psi(t)} = e^{-iLt} \ket{\psi(0)}. \]
In practice, we often (including this work) use the adjacency matrix $A$ of $\Gamma$ as the Hamiltonian of the walk, so the unitary for the state transition becomes $\exp(-iAt)$.   

A discrete-time quantum walk on $\Gamma$ can be described using the adjacency matrix $A$ by making $A$ into a stochastic matrix \cite{szegedy2004quantum}. More precisely, we define a matrix $P$ as $P_{jk} = A_{jk} / \deg(j)$, where $\deg(j)$ is the degree of the vertex $v_j$ in $\Gamma$. If the $\Gamma$ has $N$ vertices, the discrete-time quantum walk on $\Gamma$ is defined by a unitary operator on the finite Hilbert space $\C^N \times \C^N$ as follows. Define the states
\[ \ket{\phi_j} = \frac{1}{\sqrt{\deg(j)}} \sum_{k = 1}^N \sqrt{P_{jk}} \ket{j, k}, \]
and the project and swap operators
\[ \Pi = \sum_{j = 1}^N \ket{\phi_j} \bra{\phi_j}, \quad S = \sum_{j, k = 1}^N \ket{j, k} \bra{k, j}. \]
Then, a step of the quantum walk is defined by the unitary $W = S(2\Pi - \mathds{1})$. 

\paragraph{Simulating continuous-time walks.}
It was shown by Childs \cite{childs2010relationship} and Berry, Childs, and Kothari \cite{berry2015hamiltonian} that continuous-time quantum walks can be simulated using discrete-time quantum walks. Since we will use their methods, we briefly outline the necessary notations. Given a Hamiltonian $H$ of dimension $N = 2^n$, let $\opnorm{H}_{\mathrm{max}} = \max_{i, j} \abs{H_{ij}}$. The Hilbert space on which $H$ acts is expanded from $\C^N$ to $\C^{2N} \times \C^{2N}$ by adding an ancilla qubit in the state $\ket{0}$ and duplicating the entire new space. To define the discrete-time walk operator, we begin by defining the orthonormal set of states
\begin{align}
    \ket{\phi_{j0}} & := \frac{1}{\sqrt{d}} \sum_{\ell \in F_j} \ket{\ell} \left( \sqrt{\frac{H_{j\ell}^*}{K}} \ket{0} + \sqrt{1 - \frac{\abs{H_{j\ell}^*}}{K}} \ket{1} \right), \label{eq:walk-sim} \\
    \ket{\phi_{j1}} & := \ket{0} \ket{1} \nonumber,
\end{align}
where $F_j$ is the set of indices corresponding to the nonzero elements in column $j$ of $H$, and $K \ge \opnorm{H}_{\mathrm{max}}$ is a constant. 

Based on the states $\ket{\phi_{jb}}$, we define an isometry $T: \C^{2N} \to \C^{2N} \times \C^{2N}$ by
\begin{equation}
    \label{eq:walk-isom}
    T := \sum_{j = 0}^{N - 1} \sum_{b \in \{0, 1\}} (\ket{j}\bra{j} \otimes \ket{b}\bra{b}) \otimes \ket{\phi_{jb}}.
\end{equation}
The discrete-time quantum walk is then defined by $W = iS(2TT^* - \mathds{1})$, where $S$ is the swap operator acting as $S \ket{j_1} \ket{b_1} \ket{j_2} \ket{b_2} = \ket{j_2} \ket{b_2} \ket{j_1} \ket{b_1}$ for all $0 \le j_1, j_2 \le N - 1$ and $b_1, b_2 \in \{0, 1\}$. To efficiently simulate the continuous-time quantum walk defined by the Hamiltonian $H$, we need to be able to efficiently apply the isometry $T$ and its inverse $T^*$, as well as the walk operator $W$. Assuming black-box query access to $H$, we have the following:
\begin{theorem}[{\cite[Theorem 1]{berry2015hamiltonian}}]
    \label{thm:sparse-sim}
    Let $H$ be a $d$-sparse Hamiltonian $H$ acting on $n$ qubits. Then, the unitary $e^{-iHt}$ can be approximated to within error $\epsilon$ using
    \[ O\left( \tau \frac{\log(\tau / \epsilon)}{\log\log(\tau / \epsilon)} \right) \]
    queries to $H$ and
    \[ O\left( \tau [n + \log^{5 / 2}(\tau / \epsilon)] \frac{\log(\tau / \epsilon)}{\log\log(\tau / \epsilon)} \right) \]
    additional elementary gates, where $\tau = d \opnorm{H}_{\mathrm{max}} t$.
\end{theorem}


\subsection{Group Action Quantum Walks}
\label{sec:ga-qwalk}

Let $G$ be an abelian group and let $Q = \{q_1, q_2, \dots, q_k\} \subset G$ be a symmetric set, i.e., $q \in Q$ if and only if $-q \in Q$. The Cayley graph associated to $G$ and $Q$ is a graph $\Gamma = (V, E)$, where the vertex set is $V = G$, and the edge set $E$ consists of pairs $(a, b) \in G \times G$ such that there exists $q \in Q$ with $b = q + a$. The adjacency matrix of $\Gamma$ can be expressed as
\[ A = \sum_{a \in G} \lambda_a \ket{\hat{a}} \bra{\hat{a}}, \]
where $\ket{\hat{a}}$ is the quantum Fourier transform of $\ket{a}$. The eigenvalues $\lambda$ are given by
\[ \lambda_a = \sum_{q \in Q} \chi(a, q). \]
Note that the eigenvectors $\ket{\hat{a}}$ of $A$ depend only on $G$ and not on the set $Q$. 

Cayley graphs can also be constructed using group actions. Given a regular group action $(G, X, *)$ with a fixed element $x \in X$ and a set  $Q = \{q_1, q_2, \dots, q_k\} \subset G$, let $\Gamma = (X, E)$ be a graph with vertex set $X$ and edge set consisting of pairs $(x, y) \in X \times X$ such that $y = q * x$ for some $q \in Q$. The adjacency matrix of $\Gamma$ is
\[ A = \sum_{h \in G} \lambda_h \ket{G^{(h)} * x} \bra{G^{(h)} * x}, \]
where $\lambda_h = \sum_{q \in Q} \chi(h, q)$. Again, the eigenvectors $\ket{G^{(h)} * x}$ depend only on $G$. This construction of Cayley graphs from group actions generalizes the previous construction. Specifically, if we set $X = G$ and the action $*$ as group operation, we recover the original construction.

Since the action $(G, X, *)$ is regular, the two constructions yield the same graph up to isomorphism. In the first graph, the vertex set is $G$, and the rows and columns of the adjacency matrix are indexed by the elements of $G$, whereas in the second graph, the vertex set is $X$, and the rows and columns of the adjacency matrix are indexed by the elements of $X$. The isomorphism between the two graphs is induced by the bijection
\[
    \setlength{\arraycolsep}{1mm}
    \begin{array}{rlll}
        \phi: & G & \longrightarrow & X \\
        & g & \longmapsto & g * x.
    \end{array}
\]

Assuming the group action is a cryptographic group action, the isomorphism $\phi$ is a one-way function: given $g \in G$, it is easy to compute $g * x$, whereas given $(x, g * x)$, it is hard to compute $g$. Therefore, although these graphs are mathematically the same, computational problems based on them require fundamentally different techniques. A notable example is the implementation of quantum walks on these graphs.\footnote{Implementations of continuous-time quantum walks in the context of cryptographic group actions have previously appeared in \cite{booher2024failing, doliskani2023sample}, in the setting of supersingular isogeny graphs. These walks can be interpreted as group-action walks.} From this point onward, assume $Q$ is a small set, i.e., $\abs{Q} = \poly(\log \abs{G})$. In the first graph, we can efficiently implement the walks $e^{-iAt}$ even for values of $t$ that are exponentially large in $\log \abs{G}$. This follows from the fact that
\begin{equation}
    \label{eq:adjmat-fb}
    e^{-iAt} = \qft_G \sum_{a \in G} e^{-i\lambda_a t} \ket{a} \bra{a} \qft_G^*.
\end{equation}
The quantum Fourier transform $\qft_G$ and its inverse can be applied in $\poly(\log \abs{G})$ operations. The diagonal unitary in the middle, which is a phase computation $\ket{a} \mapsto e^{-i\lambda_a t} \ket{a}$, is also efficient because $\lambda_a$ can be computed classically to arbitrary precision in $\poly(\log \abs{G})$ time.

The situation for group actions is less straightforward. While the states $\ket{G^{(h)} * x}$ are analogous to the states $\ket{\hat{h}}$, the involvement of the action $*$ in the former introduces computational challenges. Specifically, it becomes hard to apply transformations beyond the action $\ket{y} \mapsto \ket{a * y}$ for $a \in G$. Consequently, no decomposition analogous to \eqref{eq:adjmat-fb} exists in the context of group actions. Instead, we are left with the spectral expression
\[ e^{-iAt} = \sum_{h \in G} e^{-i\lambda_h t} \ket{G^{(h)} * x} \bra{G^{(h)} * x}. \]
Despite this limitation, the sparsity and structure of the matrix $A$ allow us to demonstrate in the next section that for $t = \poly(\log \abs{G})$, the walk $e^{-iAt}$ can still be efficiently approximated to polynomial accuracy.   

\paragraph{Simulating group action quantum walks.}
Assume $t = \poly(\log\abs{G})$. We show that the walk $W = e^{-iAt}$ can be efficiently simulated using the discrete-time quantum walk technique from \cite{childs2010relationship, berry2015hamiltonian}. Consider the isometry $T$ in \eqref{eq:walk-isom}, where the states $\ket{\phi_{b}}$ are defined in \eqref{eq:walk-sim}. It suffices to show that $T$ and its inverse $T^*$, and the discrete-time walk $W = iS(2TT^* - \mathds{1})$ can be implemented efficiently for group actions. Note that for the Cayley graph $\Gamma = (X, E)$ of the group action $(G, X, *)$, the Hamiltonian $H$ is the adjacency matrix $A$ of $\Gamma$, which is indexed by the elements of $X$. The isometry $T$ then becomes
\begin{align*}
    T
    & := \sum_{y \in X} \sum_{b \in \{0, 1\}} (\ket{y}\bra{y} \otimes \ket{b}\bra{b}) \otimes \ket{\phi_{yb}} \\
    & = \sum_{y \in X} \sum_{b \in \{0, 1\}} \ket{y} \ket{b} \ket{\phi_{yb}} \bra{y} \bra{b}.
\end{align*}

Because of the structure of $\Gamma$, the states \eqref{eq:walk-sim} simplify as follows. First, since the set $Q$ is symmetric, the graph $\Gamma$ is undirected and the adjacency matrix $A$ is symmetric with non-negative entries; the nonzero entries of $A$ are all equal to $1$. Therefore, we can set $K = \opnorm{A}_{\mathrm{max}} = 1$. As a result, for $y \in X$, the state $\ket{\phi_{y0}}$ becomes
\[ \ket{\phi_{y0}} = \frac{1}{\sqrt{\abs{Q}}} \sum_{q \in Q} \ket{q * y} \ket{0}. \]

We now show how to implement $T$ by constructing efficient unitaries $U_0$ and $U_1$ defined as $U_b: \ket{0} \ket{y, b} \ket{0, 0} \mapsto \ket{0} \ket{y, b} \ket{\phi_{yb}}$, $b = 0, 1$, where the first register is an ancilla. Since $\ket{\phi_{y1}} = \ket{0} \ket{1}$, the unitary $U_1$ can be implemented efficiently. The unitary $U_0$ can be implemented as follows. Since $Q$ is a small set, we can construct an efficient unitary $V_Q: \C^Q \to \C^Q$ such that $V_Q \ket{0} = \abs{Q}^{-1 / 2} \sum_{q \in Q} \ket{q}$. Applying the unitary $V_Q \otimes \mathds{1}$ to the state $\ket{0} \ket{y, 0} \ket{0, 0}$ results in 
\[ \frac{1}{\sqrt{\abs{Q}}} \sum_{q \in Q} \ket{q} \ket{y, 0} \ket{0, 0}. \]
Next, we apply the unitary $V_1: \ket{q} \ket{y, 0} \ket{0, 0} \mapsto \ket{q} \ket{y, 0} \ket{q * y, 0}$. Following this, we uncompute the first register to obtain $\ket{0} \ket{y, 0} \ket{\phi_{y0}}$. Uncomputing the first register involves applying the unitary operation $V_2: \ket{q} \ket{y, 0} \ket{q * y, 0} \mapsto \ket{0} \ket{y, 0} \ket{q * y, 0}$, which can be implemented efficiently since $Q$ is small and we can recover $q$ from the pair $(y, q * y)$ by checking against all the elements in $Q$.

To implement $U_0^*: \ket{0} \ket{y, 0} \ket{\phi_{y0}} \mapsto \ket{0} \ket{y, 0} \ket{0, 0}$, we first apply $V_2^*$, followed by $V_1^*$, and finally $V_Q^* \otimes \mathds{1}$. Since each of these unitaries can be applied efficiently, $U_0^*$ can also be implemented efficiently. Now, the isometry $T$ can be constructed using a conditional unitary $U_T$, which applies $U_0$ or $U_1$ based on the qubit containing $b$. 

Finally, we show how the walk unitary $W = iS(2TT^* - \mathds{1})$ can be applied efficiently. To do this, we must show that the reflection $2TT^* - \mathds{1}$ can be implemented efficiently. Observe that
\begin{align*}
    2\ket{0} \bra{0} \otimes TT^* - \mathds{1}
    & = 2\ket{0} \bra{0} \otimes \sum_{y \in X} \sum_{b \in \{0, 1\}} \ket{y, b} \ket{\phi_{yb}} \bra{y, b} \bra{\phi_{yb}} - \mathds{1} \\
    & = 2\sum_{y \in X} \sum_{b \in \{0, 1\}} \ket{0} \ket{y, b} \ket{\phi_{yb}} \bra{0} \bra{y, b} \bra{\phi_{yb}} - \mathds{1} \\
    & = U_T \Big( 2\sum_{y \in X} \sum_{b \in \{0, 1\}} \ket{0} \ket{y, b} \ket{0, 0} \bra{0} \bra{y, b} \bra{0, 0} - \mathds{1} \Big) U_T^* \\
    & = U_T (2 \ket{0} \bra{0} \otimes \mathds{1}_{X, b} \otimes \ket{0, 0} \bra{0, 0} - \mathds{1}) U_T^*.
\end{align*}
Since $U_T, U_T^*$, and $2 \ket{0} \bra{0} \otimes \mathds{1}_{X, b} \otimes \ket{0, 0} \bra{0, 0} - \mathds{1}$ can all be applied efficiently, it follows that $2\ket{0} \bra{0} \otimes TT^* - \mathds{1}$ can also be applied efficiently. Now, for any state $\ket{\psi}$,
\[ (2\ket{0} \bra{0} \otimes TT^* - \mathds{1}) \ket{0} \ket{\psi} = \ket{0} (2TT^* - \mathds{1}) \ket{\psi}. \]

\subsection{Computing The Serial Number}

Given a state $\ket{\Z_N^{(h)} * x}_H$, we show how to compute $h$ using continuous-time quantum walks. For any $u \in \Z_N$, define a Cayley graph $\Gamma = (\Z_N, E)$ with the generating set $Q = \{-u, u\}$, as described in Section \ref{sec:ga-qwalk}. Let $A$ denote the adjacency matrix of $\Gamma$. The eigenvectors and corresponding eigenvalues of $A$ are $\ket{\Z_N^{(h)} * x}$ and $\lambda_h = 2\cos(2\pi uh / N)$, respectively, for $h \in \Z_N$. According to Theorem \ref{thm:sparse-sim} and the subsequent discussion, the unitary $W = e^{iAt}$ can be efficiently simulated to exponential accuracy. We need the following lemma.

\begin{lemma}
    \label{lem:h-eigen}
    The money state $\ket{\Z_N^{(h)} * x}_H$ is an eigenstate of $W$ with eigenvalue $e^{i\lambda_h t}$.
\end{lemma}
\begin{proof}
    We have 
    \begin{align*}
        e^{iAt} \ket{\Z_N^{(h)} * x}_H
        & = \sum_{g \in \Z_N} e^{i\lambda_gt} \ket{\Z_N^{(g)} * x} \braket{\Z_N^{(g)} * x}{\Z_N^{(h)} * x}_H \\
        & = \sum_{g \in \Z_N} e^{i\lambda_gt} \ket{\Z_N^{(g)} * x} \bra{\Z_N^{(g)} * x} \Big( \frac{1 - i}{2} \ket{\Z_N^{(h)} * x} + \frac{1 + i}{2} \ket{\Z_N^{(-h)} * x} \Big) \\
        & = e^{i\lambda_{h}t} \frac{1 - i}{2} \ket{\Z_N^{(h)} * x} + \frac{1 + i}{2} e^{i\lambda_{-h}t} \ket{\Z_N^{(-h)} * x} \\
        & = e^{i\lambda_{h}t} \ket{\Z_N^{(h)} * x}_H,
    \end{align*}
    where the second equality follows from the identity in \eqref{eq:ht-ft}, and the last equality follows from the fact that $\lambda_h = \lambda_{-h}$.
\end{proof}

If we choose $t = \poly(\log N)$, it follows from Lemma \ref{lem:h-eigen} that we can run the phase estimation algorithm with the unitary $W$ and the eigenstate $\ket{\Z_N^{(h)} * x}_H$ to compute an estimate $\tilde{\lambda}_h$ of $\lambda_h$ such that $\abs{\tilde{\lambda}_h - \lambda_h} \le 1 / \poly(\log N)$. From the estimate $\tilde{\lambda}_h$ we can obtain a value $0 \le \theta \le 1$ such that
\[ \abs[\Big]{\theta - \frac{uh}{N}} \le \frac{1}{\poly(\log N)}. \]
Since the phase estimation algorithm can be executed for different values of $u$, we can obtain different estimates of $uh / N$. It was shown in \cite{zhandry2024quantum} that for carefully chosen values of $u$, these estimates provide sufficient information to fully recover $h$.

\newpage
\bibliographystyle{plain}
\bibliography{references}

\end{document}